\documentclass[11pt,twoside]{article}
\usepackage{amsmath}
\usepackage{amssymb}
\usepackage{makeidx}
\usepackage{graphicx}
\usepackage{latexsym}
\usepackage{float}
\usepackage{amssymb,amsmath,amsfonts}
\usepackage{xspace}
\usepackage{hyperref}
\usepackage{multirow}
\usepackage[usenames]{color}
\newtheorem{theorem}{Theorem}[section]
\newtheorem{lemma}[theorem]{Lemma}

\newtheorem{corollary}[theorem]{Corollary}

\newtheorem{definition}[theorem]{Definition}

\allowdisplaybreaks[3]      



\makeatletter
\floatstyle{ruled}
\newfloat{fragment}{H}{lop}
\floatname{fragment}{Algorithm}
\renewcommand{\floatc@ruled}[2]{\vspace{2pt}{\@fs@cfont #1.\:} #2 \par
  \vspace{1pt}}
\makeatother

\newcommand{\mypseudocodelabel}[1]{\hfil}

\DeclareMathSymbol{\qedsymb} {\mathord}{AMSa}{"04}

\newcommand{\ceil}[1]{\left\lceil #1 \right\rceil}
\newcommand{\floor}[1]{\left\lfloor #1 \right\rfloor}

\newcommand{\pred}{\mbox{pred}}
\newcommand{\dist}{\mbox{dist}}

\newcommand{\val}{\mathrm{val}}

\newcommand{\varass}{\leftarrow}

\newcommand{\CorollaryName}[1]{\label{cor:#1}}
\newcommand{\DefinitionName}[1]{\label{def:#1}}

\newcommand{\LemmaName}[1]{\label{lem:#1}}

\newcommand{\SectionName}[1]{\label{sec:#1}}
\newcommand{\TheoremName}[1]{\label{thm:#1}}
\newcommand{\FigureName}[1]{\label{fig:#1}}

\newcommand{\Corollary}[1]{Corollary~\ref{cor:#1}}
\newcommand{\Definition}[1]{Definition~\ref{def:#1}}

\newcommand{\Lemma}[1]{Lemma~\ref{lem:#1}}

\newcommand{\Theorem}[1]{Theorem~\ref{thm:#1}}
\newcommand{\Figure}[1]{Figure~\ref{fig:#1}}

\newcommand{\proofbelow}{8pt}
\newcommand{\afterproof}{\hfill $\blacksquare$ \par \vspace{\proofbelow}}
\newcommand{\aftersubproof}{\hfill $\Box$ \par \vspace{\proofbelow}}
\newenvironment{proof}{\noindent\textbf{Proof.}\,}{\afterproof}

\newenvironment{proofof}[1]{\noindent\textit{Proof} \,(of #1).\,}{\afterproof}

\renewcommand{\th}{\ifmmode{^{\textrm{th}}}\else{\textsuperscript{th}\ }\fi}

\newcommand{\tour}{\texttt{SampledTournament}\xspace}
\newcommand{\kmaxfind}[1]{$#1$-{\texttt{MaxFind}\xspace}}
\newcommand{\ksort}[1]{$#1$-{\texttt{Sort}\xspace}}
\newcommand{\kpivot}[1]{$#1$-\texttt{Pivot}\xspace}
\newcommand{\kselect}[1]{$#1$-\texttt{Select}\xspace}
\newcommand{\onecover}{$1$-\texttt{Cover}\xspace}

\newcommand{\roundRobin}{round-robin }

\newcommand{\ignore}[1]{\relax}
\newcommand{\alequ}[1]{\begin{align} #1 \end{align}}
\newcommand{\alequn}[1]{\begin{align*} #1 \end{align*}}

\newcommand{\outd}{d_{\mbox{out}}}

\renewcommand{\thefootnote}{\fnsymbol{footnote}}
\title{Sorting and Selection with Imprecise Comparisons\footnote{A preliminary version of this work containing weaker forms of some of the results has appeared in the proceedings of 36th International Colloquium on Automata, Languages and Programming (ICALP) 2009.}}
\author{Mikl\'{o}s Ajtai\\ IBM Research - Almaden \\ {\tt ajtai@us.ibm.com} \and Vitaly Feldman \\ IBM Research - Almaden \\{\tt vitaly@post.harvad.edu} \and
  Avinatan Hassidim\thanks{Part of this work was done while the author was at Google Research, Israel.} \\ Bar Ilan University, Israel \\{\tt avinatanh@gmail.com}, \and Jelani Nelson\thanks{Supported by NSF grant CCF-0832797. Part of this work was done while the author was at IBM Research - Almaden.} \\ Harvard University \\{\tt minilek@seas.harvard.edu}}

\begin{document}

\pagestyle{plain}

\maketitle

\renewcommand{\thefootnote}{\arabic{footnote}}
\begin{abstract}
We consider a simple model of imprecise comparisons: there exists
some $\delta>0$ such that when a
subject is given two elements to compare, if the values of those elements
(as perceived by the subject)
differ by at least $\delta$, then the comparison will be made
correctly; when the two elements have values
that are within $\delta$, the outcome of the comparison is
unpredictable. This model is inspired by both imprecision in human judgment of values and also by bounded but potentially adversarial errors 
in the outcomes of sporting tournaments.

Our model is closely related to a number of models commonly considered in the psychophysics literature where
$\delta$ corresponds to the {\em just noticeable difference unit
(JND)} or {\em difference threshold}. In experimental psychology, the method of paired comparisons
was proposed as a means for ranking preferences amongst $n$ elements of a human subject. The method requires performing
all $\binom{n}{2}$ comparisons, then sorting elements according to the
number of wins.  The large number of comparisons is performed
to counter the potentially faulty decision-making of the human
subject, who acts as an imprecise comparator. 

We show that in our model the method of paired comparisons has optimal
accuracy, minimizing the errors introduced by the imprecise
comparisons. However, it is also wasteful, as it requires all $\binom{n}{2}$.
We show that the same optimal
guarantees can be achieved using $4 n^{3/2}$ comparisons, and we
prove the optimality of our method. We then explore the general tradeoff
between the guarantees on the error that can be made and number of
comparisons for the problems of sorting, max-finding, and
selection. Our results provide strong lower bounds and close-to-optimal solutions for each of
these problems.
\end{abstract}

\section{Introduction}
Let $x_1, \ldots, x_n$ be $n$ elements where each $x_i$ has
an unknown value $\val(x_i)$. We want to find the element with the
maximum value using only pairwise
comparisons. However, the outcomes of comparisons are imprecise in the
following sense.
For some fixed $\delta >0$, if $|\val(x_i) - \val(x_j)| \le \delta$,
then the result of the comparison can be either ``$\geq$" or
``$\leq$''. Otherwise, the result of the comparison is correct. It is
easy to see that in such a setting it might be impossible to find the
true maximum (for example when the values of all the elements are
within $\delta$). It might however be possible to identify an
approximate
maximum, that is an element $x_{i^*}$ such that for all $x_i$, $\val(x_i)
- \val(x_{i^*}) \le k \delta$ for some, preferably small, value
$k$. In addition, our goal is to minimize the number of comparisons
performed to find $x_{i^*}$. We refer to the minimum value $k$ such
that an algorithm's output is always guaranteed to be $k\delta$-close
to the maximum as the {\em error}
of the algorithm in this setting. Similarly, to sort the above
elements with error $k$ we need to find a permutation $\pi$ such that
if $\pi(i) < \pi(j)$ then $\val(x_i) - \val(x_j) \le k
\delta$.

A key issue that our work addresses is that in any sorting (or
max-finding) algorithm, errors resulting from imprecise comparisons
might accumulate, causing the final output to have high
error. Consider, for example, applying the classical bubble sort
algorithm to a list of elements that are originally sorted in the
reverse order and where the difference between two adjacent elements is
exactly $\delta$. All the comparisons will be between elements within
$\delta$ and therefore, in the worst case, the order will not be
modified by the sorting, yielding error $(n-1)\delta$. Numerous other known
algorithms that primarily optimize the number of comparisons can be
easily shown to incur a relatively high error. As can be easily
demonstrated (\Theorem{roundrobin}), performing all $\binom{n}{2}$
comparisons then sorting elements according to the number of wins, a
``round-robin tournament'',
achieves error $k= 2$, which is lowest possible
(\Theorem{performance-2}). A natural question we ask here is whether
$\binom{n}{2}$ comparisons are necessary to achieve the same error. We
explore the same question for all
values of $k$ in the problems of sorting, max-finding, and general
selection.

One motivation for studying this problem comes from social
sciences. A common problem both in experimental psychology and
sociology is to have
a human subject rank preferences amongst many candidate options.  It
also occurs frequently in marketing research \cite[Chapter 10]{SA05},
and in training information retrieval algorithms using human
evaluators \cite[Section 2.2]{AAC+08}.
The basic method to elicit preferences is to present the
subject two alternatives at a time and ask which is the
preferred one. The common approach to this problem today was presented
by Thurstone as early as 1927, and is called the ``method of paired
comparisons'' (see \cite{David88} for a thorough treatment). In this
method, one asks the
subject to give preferences for all
pairwise comparisons amongst $n$ elements. A ranked preference list is
then determined by the number of ``wins" each candidate element
receives. A central concept in these studies introduced as far back as
the 1800s by Weber and Fechner is that of the {\em just noticeable
  difference (JND) } unit or {\em difference threshold} $\Delta$.  If
two physical stimuli with intensities $x \leq y$ have $y \le x + \Delta x$, a
human will not be able to reliably distinguish which intensity is
greater. The idea was later generalized by Thurstone to
having humans not only compare physical stimuli, but also
abstract concepts \cite{Thurstone27a}. By the Weber-Fechner law, stimuli with intensities $x$ and $y$ cannot be distinguished when $1/(1+\Delta) \leq x/y \leq 1+\Delta$. This is equivalent to saying that intensities are indistinguishable when the absolute difference between the natural logarithms of the intensities is less than $\delta = \ln(1+\Delta)$. Therefore JND in the Weber-Fechner law corresponds to the imprecision $\delta$ of our model of comparisons when intensities are measured on the logarithmic scale. More generally, we can always assume that the intensities are measured on the scale for which JND corresponds to an absolute difference between values.

Most previous work on the method of paired comparisons has been
through the lens of statistics. In such work the JND is modeled as a
random variable and the statistical properties of Thurstone's method
are studied \cite{David88}. Our problem corresponds to a simplified model of
this problem which does not require any statistical assumptions and
is primarily from a combinatorial perspective.

Another context that captures the intuition of our model is that of
designing a sporting tournament based on win/lose games.
There, biases of a judge and unpredictable events can change the
outcome of a game when the strengths of the players are close. Hence
one cannot necessarily assume that the outcome is truly random in such
a close call. It is clear that both restricting the influence of
the faulty outcomes and reducing the total number of games required
are important in this scenario, and hence exploring the tradeoff
between the two is of interest.
For convenience, in the rest of the paper we often use the terminology
borrowed from sporting tournaments.

Accordingly, the problems we consider have a natural interpretation as problems on a tournament graph (that is, a complete directed graph with only one edge between any two vertices). We can view all the comparisons that were revealed by the comparator as a digraph $G$. The vertices of $G$ are the $n$ elements and it
contains the directed edge $(x_i, x_j)$ if and only if a comparison between $x_i$ and $x_j$ has been made, and the comparator has responded with ``$x_i \geq x_j$''. At any point in time the comparison graph is a subgraph of some unknown tournament graph. The problem of finding a maximum element with error $k$ is then equivalent to finding a vertex in such a graph from which there exists a directed path of length at most $k$ to any other vertex while minimizing the number of edges which need to be revealed (or ``matches" in the context of tournaments). Such vertex is referred to as a $k$-{\em king} (or just king for $k=2$) \cite{Landau53,Maurer80}. Existence and properties of such elements for various tournament graphs have been studied in many contexts. Sorting with error $k$ gives an ordering of vertices such that if vertex $x_i$ occurs after $x_j$ in the order then there exists a directed path of length at most $k$ from $x_i$ to $x_j$. The connection to tournament graphs is made explicit in Section \ref{sec:lb}.

Finally, in a number of theoretical contexts responses are given by an imprecise oracle. For example, for weak oracles given by Lov\'{a}sz in the context of optimization \cite{Lovasz:86} and for the statistical query oracle in learning \cite{Kearns:98} the answer of the oracle is undefined when some underlying value $z$ is within a certain small range of the decision boundary.
When $z$ itself is the difference of two other values, say $z_1$ and $z_2$, then oracle's answer is, in a way, an imprecise comparison of $z_1$ and $z_2$. This correspondence together
with our error 2 sorting algorithm was used by one of the authors to derive algorithms in the context of evolvability \cite{Feldman:09robust}.

\subsection{Our results}
We first examine the simpler problem of finding only the maximum
element.
For this problem, we give a deterministic
max-finding algorithm with error $2$
using $2n^{3/2}$ comparisons. This contrasts with the method of
paired comparisons, which makes $(n^2-n)/2$ comparisons to achieve
the same error. Using our algorithm recursively,
we build deterministic algorithms with error $k$
that require $O(n^{1 + 1/(3 \cdot 2^{k-2} - 1)})$ comparisons.
We also give a lower bound of $\Omega(n^{1 +  1/(2^k - 1)})$ comparisons for the problem.
The bounds are almost tight --- the upper bound for our error $k$
algorithm is less than our lower bound for error $(k-1)$ algorithms.
We also give a linear-time randomized algorithm
that achieves error $3$ with probability at least $1 - 1/n^2$, showing
that randomization greatly changes the complexity of the problem.

We then study the problems of selecting an element of a certain order and sorting. For $k=2$, we give a deterministic algorithm that sorts using $4 n^{3/2}$ comparisons (and in particular can be used for selecting an element of any order). For general $k$, we show that selection of an element of any order $i$ can be achieved using $O(2^k \cdot n^{1 + 1/2^{k-1}})$ comparisons
and sorting with error $k$ can performed using $O(4^k \cdot n^{1 + 1/2^{k-1}})$ comparisons.

We give a lower bound of $\Omega(n^{1  + 1/2^{k-1}})$ comparisons for sorting with error $k$. When $k=O(1)$ our bounds are tight up to a constant factor and are at most a $\log n$ factor off for general $k$. Our lower bounds for selection depend on the order of the element that needs to be selected and interpolate between the lower bounds for max-finding and the lower bounds for
sorting.  For $k \ge 3$, our lower bound for finding the median (and also for sorting) is
strictly larger than our upper bound for max-finding. For example, for $k=3$ the
lower bound for sorting is $\Omega(n^{5/4})$, whereas max-finding
requires only $O(n^{6/5})$ comparisons.

Note that we achieve $\log\log n$ error for max-finding in $O(n)$
comparisons, and $\log\log n$ error for sorting in $O(n\log^2 n)$ comparisons. Standard methods
using the same number (up to a $\log n$ factor)  of comparisons (e.g. a single-elimination
tournament tree, or Mergesort) can be shown to incur error at least $\log n$.
Also, all the algorithms we give are efficient in that their running
times are of the same order as the number of comparisons they make.

The basis of our deterministic algorithms for both max-finding
and selection are efficient algorithms for
a small value of $k$ ($k=2$). The algorithms for larger error $k$ use several different ways to
partition elements, then recursively use algorithms for smaller error and then
combine results. Achieving nearly tight results for max-finding requires in part
relaxing the problem to that of finding a small
{\em $k$-max-set}, or a set which is guaranteed to contain at least
one element of value at least $x^* - k\delta$, where $x^*$ is the
maximum value of an element (we interchangeably use $x^*$ to refer
to an element of maximum value as well).  It turns out
we can find a $k$-max-set in a fewer number of comparisons than the
lower bound for error-$k$ max-finding algorithms.
Exploiting this allows us to develop an efficient recursive
max-finding algorithm.
We note a similar approach of finding a small set of ``good'' elements
was used by Borgstrom and Kosaraju \cite{BorgKo93} in the context of
noisy binary search.

To obtain our lower bounds for deterministic algorithms we show that the problems we consider have equivalent formulations on tournament graphs in
which the goal is to ensure existence of short (directed) paths from a certain node to other nodes.
 Using a comparison oracle that always prefers elements that had fewer wins in previous
rounds, we obtain bounds on the minimum of edges that are required to
create the paths of desired length. Such bounds are then translated
back into bounds on the number of comparisons required to achieve
specific error guarantees for the problems we consider.

\begin{table}
\begin{tabular}{|| c || c | c | c | c ||}
\hline
\hline
\multirow{2}{*}{Task} & \multicolumn{3}{|c|}{Upper bounds} & \multirow{2}{*}{Lower bound} \\
\cline{2-4}
&  $k=2$ & $k$ & $k = \log\log n$ & \\
\hline
Find maximum & $2 n^{3/2}$  & $O(n^{1 + 1/(3 \cdot 2^{k-2} - 1)})$ & $O(n)$ & $\Omega(n^{1 +  1/(2^k - 1)})$ \\
\hline
Select $i$-th & $4 n^{3/2}$  & $O(2^k \cdot n^{1 + 1/2^{k-1}})$ & $O(n \log n)$ & $\Omega(i \cdot \max\{i^{1/(2^k-1)},
n^{1/(2^{k-1})}\})$ \\
\hline
Sort & $4 n^{3/2}$  & $O(4^k \cdot n^{1 + 1/2^{k-1}})$ & $O(n \log^2 n)$ & $\Omega(n^{1  + 1/2^{k-1}})$ \\
\hline
\hline
\end{tabular}
\caption{\small Overview of the bounds for deterministic algorithms with error $k$. In the selection task $i$-th smallest element is chosen (with maximum being $n$-th smallest). \label{fig:table}}
\end{table}

For our randomized max-finding algorithm, we use a type of
tournament with random seeds at each level, in combination with random
sampling at each level of the tournament tree.  By performing a
\roundRobin tournament on the
top few tournament players together with the sampled elements, we
obtain an element of value at least $x^*-3\delta$ with polynomially
small error probability.

\subsection{Related Work}

Handling noise in binary search procedures was first considered by
R\'enyi \cite{Ren61} and by Ulam \cite{ulam}.  An algorithm for solving
Ulam-R\'enyi's game was proposed by Rivest et al.~\cite{KMRSW80}, where an adversarial comparator can err a bounded
number of times. They
gave an algorithm with
query complexity $O(\log n)$ which succeeds if the number of
adversarial errors is constant.

Yao and Yao \cite{yao1985ftn}
introduced the problem of sorting and of finding the maximal element
in a sorting network
when each comparison
gate either returns the right answer or does not work at all.  For
finding the maximal element, they showed that it is necessary and
sufficient to use $(e+1)(n-1)$ comparators when $e$ comparators can be
faulty. Ravikumar, Ganesan and Lakshmanan extended the model to
arbitrary errors, showing that  $O(en)$ comparisons are necessary and
sufficient \cite{ravikumar1987sle}. For sorting, Yao and Yao showed
that $O(n \log n + en)$ gates are sufficient. In a different fault
model, and with a different definition of a successful sort, Finocchi
and Italiano \cite{finocchi2004sas} showed an $O(n \log n)$ time
algorithm resilient to $(n \log n)^{1/3}$ faults. An improved
algorithm handling $(n \log n)^{1/2}$ faults was later given
by Finocchi, Grandoni and Italiano \cite{FinGraIta09}.

In the model where each comparison is incorrect with some
probability $p$, Feige et al.\ \cite{FeigeEtAl94} and Assaf and Upfal
\cite{AU91} give algorithms for
several comparison problems, and \cite{BH08,KK07} give algorithms for
binary search.
We refer the reader interested in the rich history and models
of faulty comparison problems to a survey of Pelc
\cite{Pelc02} and a monograph of Cicalese \cite{Cicalese13}.

We point out that some of the bounds we obtain appear similar to those
known for max-finding, selection, and sorting in parallel in
Valiant's model \cite{Valiant75}. In particular, our bounds for max-finding are close to those obtained by Valiant for the parallel analogue of the problem (with the error used in place of parallel time) \cite{Valiant75}, and our lower bound of $\Omega(n^{1 +
  1/(2^k - 1)})$ for max-finding with error $k$ is identical to a lower (and upper) bound given by H\"{a}ggkvist and Hell
\cite{HH82} for merging two sorted arrays each of length $n$ using a
$k$-round parallel algorithm. Despite these similarities in bounds, our
techniques are different, and
we are not aware of any deep connections. As some evidence of the
difference between the problems we note that for sorting in $k$
parallel rounds it is known that $\Omega(n^{1+1/k})$ comparisons are
required \cite{AA88,BT83,HH81}, whereas in our
model, for constant $k$, we can sort with error $k$ in
$n^{1+1/2^{\Theta(k)}}$ comparisons. For a survey on
parallel sorting algorithms, the reader is
referred to \cite{GGK03}.

The authors have recently learned that finding a king and sorting kings in a tournament graph while minimizing the number of uncovered edges was previously studied by Shen, Sheng and Wu \cite{ShenSW03}. As follows from our results, this problem is equivalent to max-finding and sorting by a deterministic algorithm with error 2. Their upper and lower bounds for this case are (asymptotically) identical to our bounds and are based on essentially the same techniques. Our results can be seen as a generalization of their results to $k$-kings for all $k \geq 2$. We also remark that for randomized algorithms our problem is no longer equivalent to the problem considered in \cite{ShenSW03}.

\section{Notation}
 Throughout this document we let $x^*$ denote
some $x_i$ of the maximum value (if there are several such elements, we
choose one arbitrarily).  Furthermore, we use $x_i$ interchangeably to
refer to the both the $i^{th}$ element and its value, e.g. $x_i > x_j$
should be
interpreted as $\val(x_i) > \val(x_j)$.

We assume $\delta = 1$ without loss of generality, since the problem with
arbitrary $\delta>0$ is equivalent to the problem with $\delta=1$ and
input values $x_i/\delta$. We stress that the algorithm does not know
$\delta$.

We say {\em $x$  defeats $y$} when the comparator
claims that $x$ is larger than $y$ (and we similarly use the phrase
{\em $y$ loses to $x$}). Note that $x$ defeats $y$ implies $x \ge y - 1$. We do not necessarily assume that repeating the same comparison several times would give the same result, and our algorithms do not repeat a comparison twice.
We say $x$ is {\em $k$-greater} than $y$ ($x \ge_k y$) if
$x \ge y - k$.  The term {\em $k$-smaller} is defined analogously. A set of elements $T$ is $k$-greater than a set of elements $S$ if for every $y\in S$ and every $x\in T$, $x \geq_k y$. We say an element is a {\em $k$-max} of a set if it is $k$-greater than all other elements in
the set. If the set is not specified explicitly then we refer to the set of all input elements.
A permutation $x_{\pi(1)},\ldots,x_{\pi(n)}$ is {\em $k$-sorted} if $x_{\pi(i)}
\ge_k x_{\pi(j)}$ for every $i>j$.  A {\em $k$-max-set} is a subset
of all elements which contains at least
one element of value at least $x^* - k$.

All logarithms throughout this document are base-$2$. For simplicity
of presentation, we occasionally omit floors and ceilings and ignore rounding errors when they have an insignificant effect on the bounds.

\section{Max-Finding}\SectionName{maxfind}
In this section we give deterministic and randomized algorithms for
max-finding.

\subsection{Deterministic Algorithms}
We start by showing that
the method of paired comparisons provides an optimal error guarantee,
not just for max-finding, but also for sorting.

\begin{theorem}\TheoremName{roundrobin}
Sorting according to the number of wins in a \roundRobin tournament
has error at most $2$.
\end{theorem}
\begin{proof}
Let $x,y$ be arbitrary elements with $y$
strictly less than $x - 2$.  For any $z$ that $y$ defeats, $x$ also
defeats $z$.  Furthermore, $x$ defeats $y$, and thus $x$ has strictly
more wins than $y$, implying $y$ is placed lower in the sorted order.
\end{proof}

\begin{theorem}\TheoremName{performance-2}
No deterministic max-finding algorithm has error less than $2$.
\end{theorem}
\begin{proof}
Given three elements $a,b,c$, the comparator can claim $a>b>c>a$, making
the elements indistinguishable. Without loss of generality, suppose $A$
outputs $a$.  Then the values could be
$a=0$, $b=1$, $c=2$, implying $A$ has error $2$.
\end{proof}

In \Figure{A2} we give our error 2 algorithm for max-finding.
\begin{figure*}
\begin{center}
\fbox{
\parbox{6in} {
\underline{Algorithm \kmaxfind{2}}:
\texttt{// Returns an element of value at least $x^*-2$.
  The value $s>1$ is a parameter which is by default $\ceil{\sqrt{n}}$
  when not specified.
  }
\begin{enumerate}
\item Label all $x_i$ as candidates.
\item \textbf{while} there are more than $s$ candidate
  elements:
\begin{enumerate}
\item Pick an arbitrary subset of $s$ of the candidate elements and
  play them in a \roundRobin tournament.  Let $x$ be the element with the largest number
  of wins.
\item Compare $x$ against all candidate elements and eliminate all elements that
  lose to $x$.
\end{enumerate}
\item Play the remaining (at most $s$) candidate elements in a
  \roundRobin tournament and return the element with the largest number of wins.
\end{enumerate}
}}
\end{center}
\caption{The algorithm \kmaxfind{2} for finding a $2$-max.}\FigureName{A2}
\end{figure*}

\begin{lemma}\LemmaName{A2}
For every $s \leq n$, the max-finding algorithm \kmaxfind{2} has error $2$ and makes at most $(n-s)s + (n^2-s^2)/(s-1)$
comparisons. In particular, the number of comparisons is at most $2n^{3/2}$
for $s = \ceil{\sqrt{n}}$.
\end{lemma}
\begin{proof}
We first analyze the number of comparisons.  In the $t^{th}$ iteration,
the number of comparisons is at most $\binom{s}{2} + (n_t - s)$, where $n_t$ is the number of candidate elements in round $t$.  We now bound the number of iterations and $n_t$. In all but the
last iteration, the total number of comparisons made in the \roundRobin tournament is
$\binom{s}{2} = s(s-1)/2$. Thus by an averaging argument, the
element which won the largest number of comparisons won at least
$(s-1)/2$ times.  Thus, at least $(s-1)/2$ elements are eliminated in each
iteration, implying the number of iterations is at most $2(n-s)/(s-1)$ and $n_t \leq n - t(s-1)/2$.
The total number of comparisons is thus at most $$\sum_{t \leq 2 (n-s)/(s-1)} [s(s-1)/2 + n - t(s-1)/2] \leq (n-s) s + (n^2-s^2)/(s-1)\ .$$

We now analyze error in two cases.  The first case is that $x^*$
is never eliminated, and thus
$x^*$ participates in Step 3. \Theorem{roundrobin} then ensures
that the final output is of value at least $x^*-2$.  Otherwise,
consider the iteration when $x^*$ is
eliminated.  In this iteration, it must
be the case that the $x$ chosen in Step 2(b) has $x \ge x^* - 1$, and
thus any element with value less
than $x^*-2$ was also eliminated in this iteration.  In this case all
future iterations only contain elements of value at least $x^*-2$, and so
again the final output has value at least $x^*-2$.
\end{proof}

The key recursion step of our general error max-finding algorithm is the algorithm \onecover (given in \Lemma{A2-cover}) which is based on \kmaxfind{2} and the following lemma.
\begin{lemma}\LemmaName{log-winners}
There is a deterministic algorithm which makes $\binom{n}{2}$
comparisons and outputs a $1$-max-set of size at most $\ceil{\log
  n}$.
\end{lemma}
\begin{proof}
We build the output set in a greedy manner. Initialize $S = \emptyset$. At
each step consider the subtournament on the vertices $T$ defined to be those vertices neither
in $S$ nor defeated
by an element of $S$. An averaging
argument shows there exists an element in $T$ which wins at least half
its matches in this subtournament; add this element to $S$. Note in the next step $|T|$ decreases by a factor of at least $2$, so there are at most $\ceil{\log n}$ iterations.
Furthermore, at least one element in the final set $S$ must either have value $x^*$
or have defeated $x^*$, and thus $S$ is a $1$-max-set.
\end{proof}
We now obtain \onecover by setting
$s=\ceil{2\sqrt{n}}$ in \Figure{A2}, then returning the union
of the $x$ that were chosen in any iteration of Step 2(a), in
addition to the output of \Lemma{log-winners} on the elements in the
final tournament in Step 3.

\begin{lemma}\LemmaName{A2-cover}
There is an algorithm \onecover making at most $3 \cdot n^{3/2}$
comparisons which finds a $1$-max-set of size at most
$\sqrt{n}$ (for $n \geq 81$).
\end{lemma}
\begin{proof}
Run the algorithm \kmaxfind{2} with $s = \ceil{2\sqrt{n}}$.  Return
the set consisting of all elements that
won the \roundRobin tournament in Step 2(b) of \Figure{A2} in at least
one iteration, in addition to a size-$\ceil{\log s}$
set $1$-greater than the candidate elements which were left in Step 3 (using
\Lemma{log-winners}).  The total size of the returned set is thus
$\ceil{n/s} - 1 + \ceil{\log s} \le \ceil{\sqrt{n}/2} + \ceil{\log(\ceil{2\sqrt{n}})}-1
$.  For $n \geq 81$, this is at most $\sqrt{n}$.

To show correctness, consider the element $x^*$ of maximal value.  Either
$x^*$ was eliminated in Step 2(c) of some
  iteration, in which case the element $x$ that eliminated $x^*$ had value at
  least $x^*-1$, or $x^*$ survived until the last iteration, in which case
  the set constructed via \Lemma{log-winners} is a $1$-max-set. Finally, note that the number of comparisons is the same as the number of comparisons used by \kmaxfind{2} (with the same $s$) and therefore is less than $3n^{3/2}$.
\end{proof}

We are now ready to give our main algorithm for finding a $k$-max, shown in \Figure{Ak}.


\begin{figure*}
\begin{center}
\fbox{
\parbox{6in} {
\underline{Algorithm \kmaxfind{k}}:
\texttt{// Returns a $k$-max for $k \geq 3$}
\begin{enumerate}
\item If $n \leq 81$ return the output of \kmaxfind{2} on the input elements.
\item Equipartition the $n$ elements into sets $S_1,\ldots,S_t$ each of size $r = \max\{81, 4 \cdot n^{\frac{8}{3\cdot 2^k-4}}\}$
\item Call \onecover on each set $S_i$ to recover a set $T_i$ 1-greater than $S_i$.
\item Return \kmaxfind{(k-1)}$(\cup_{i=1}^t T_i)$ \texttt{// Recursion stops at \kmaxfind{2}.}
\end{enumerate}
}}
\end{center}
\caption{The algorithm \kmaxfind{k} for finding a $k$-max.}\FigureName{Ak}
\end{figure*}

\begin{theorem}\TheoremName{upper-bound}
For every $3 \le k \le \log\log n$, \kmaxfind{k} uses
$O(n^{1 + 1/(3\cdot 2^{k-2} - 1)})$ comparisons and finds a $k$-max.
\end{theorem}
\begin{proof}
We prove that for any $2 \le k \le \log\log n$, \kmaxfind{k} uses at most $54 \cdot n^{3 \cdot 2^k/(3\cdot 2^k-4)} = 54 \cdot n^{1 + 1/(3\cdot 2^{k-2} - 1)}$ comparisons and finds a $k$-max by induction. First, by \Lemma{A2}, it holds for \kmaxfind{2}.

We now prove the bound on the error. Let $x$ be the element returned by \kmaxfind{k}. By the inductive hypothesis, for every $y \in \cup_{i=1}^t T_i$, $x \geq_{k-1} y$. In addition, for every input element $x_j$ there exists $y \in  T_i$ for some $i$ such that $y \geq_1 x_j$. Therefore, $x \geq_k x_j$ for every $j \in [n]$.

The total number of comparisons used by \kmaxfind{k} can be bounded as follows.
\begin{itemize}
\item if $n \leq 81$ then $4 \cdot n^{3/2} \leq 36 \cdot n$. Otherwise,
\item $t=n/r$ invocations of \onecover. By \Lemma{A2-cover}, this requires at most $3 \cdot r^{3/2} \cdot n/r =  3 \sqrt{r} \cdot n$ comparisons which equals $\max \{27 \cdot n, 6 \cdot n^{\frac{3 \cdot 2^k}{3\cdot 2^k-4}} \}$. Note that $n^{\frac{3 \cdot 2^k}{3\cdot 2^k-4}} \geq n$ and therefore we can use $27 \cdot n^{\frac{3 \cdot 2^k}{3\cdot 2^k-4}}$ as an upper bound.
\item The invocation of \kmaxfind{(k-1)} on $\cup_{i=1}^t T_i$. By \Lemma{A2-cover}, the size of each $T_i$ is at most $\sqrt{r}$. Therefore, $|\cup_{i=1}^t T_i| = \sqrt{r} \cdot n/r = n/\sqrt{r} \leq n^{\frac{3 \cdot 2^k-8}{3\cdot 2^k-4}}/2$. By, the inductive assumption this invocation requires at most $$54 \cdot \left(n^{\frac{3 \cdot 2^k-8}{3\cdot 2^k-4}}/2\right)^{1 + 1/(3\cdot 2^{k-3} - 1)} \leq 27 \cdot n^{\frac{3 \cdot 2^k-8}{3\cdot 2^k-4} \cdot \frac{3 \cdot 2^{k-1}}{3\cdot 2^{k-1}-4}} =  27 \cdot n^{\frac{3 \cdot 2^k}{3\cdot 2^k-4}}\ .$$
\end{itemize}
Altogether the number of comparisons is at most $\max\{36 \cdot n, 54 \cdot n^{\frac{3 \cdot 2^k}{3\cdot 2^k-4}}\}= 54 \cdot n^{\frac{3 \cdot 2^k}{3\cdot 2^k-4}}$.
\end{proof}

\begin{corollary}\CorollaryName{linear-time-max}
There is a max-finding algorithm using $O(n)$ comparisons with error of at most $\log\log n$.
\end{corollary}

\subsection{Randomized Max-Finding}
We now show that randomization can significantly reduce the number of
comparisons required to find an approximate maximum.
Our algorithm operates correctly even
if the adversary can adaptively choose how to err when two elements
are close
(though we stress that the adversary may not change input values over
the course of an execution). In particular, the
classic randomized selection algorithm can take quadratic time in this
adversarial model since
for an input with all equal values, the adversary can
claim that the randomly chosen pivot is smaller than all other
elements. Nevertheless, even in this strong adversarial model, we show
the following.

\begin{theorem} \TheoremName{exists-rand-algo}
For any integer $c\ge 1$, there exists a randomized
algorithm which given any $n$ elements finds a $3$-max of the set with
probability at least $1 - n^{-c}$ using at most $(s-1)n + ((c+1)/(C\ln
2))^2/2\cdot n^{2/3}\ln^4 n = O(n)$ comparisons, where $s,C$ are as
defined in \Figure{tour}.
\end{theorem}
Taking $c > 1$, and using the fact that the error of our algorithm can
never be more than $n-1$,
this gives an algorithm which finds an element with
expected value at least $x^* - 4$.  The high-level idea of the
algorithm is as follows. We
randomly equipartition the elements into constant-sized sets.  In each
set we play a \roundRobin tournament and advance everyone who won more
than $3/4$ of its comparisons. As we will prove, the element with the median number of wins can win at most $3/4$
of its comparisons and hence no more than half of the elements advance.  We also randomly sample a
set of elements at each level of the tournament tree.  We show that
either (1) at some round of the tournament there is an abundance of
elements with value at
least $x^*-1$, in which case at least one such element is sampled with
high probability, or (2) $x^*$ advances as one of the top few tournament
elements with high probability.
\Figure{tour} presents the subroutine \tour for the algorithm.

\begin{figure*}
\begin{center}
\fbox{
\parbox{6in} {
\underline{Algorithm \tour}:
\texttt{// For constant
  $c$ and $n$ sufficiently large,
  returns a $1$-max-set with probability at least
  $1-n^{-c}$.}
\begin{enumerate}
\item Initialize $N_0 \leftarrow \{x_1,\ldots,x_n\}$,
    $W\leftarrow\emptyset$, $s\leftarrow 15(c+2) + 1$, $C\leftarrow
    4^4/(5e)^5$, and $i\leftarrow 0$, where $e$ is Euler's number.
\item \label{halt-condition} \textbf{if} $|N_i| \le
  ((c+1)/C)n^{1/3}\ln n$, add elements of $N_i$ to
  $W$ and \textbf{return} $W$.
\item \textbf{else} add a random subset of $((c+1)/C) n^{1/3}\ln n$
  elements from $N_i$ to $W$.
\item \label{loser} Randomly partition the elements in $N_i$ into sets
  of size $s$. In each set, perform a \roundRobin tournament.
\item \label{last-step} Let $N_{i+1}$ contain all elements of $N_i$ which had strictly
  fewer than $(s-2)/4$ losses in their \roundRobin tournament in Step
  \ref{loser}. Increment $i$ and \textbf{goto} Step
  \ref{halt-condition}.
\end{enumerate}
}}
\end{center}
\caption{The algorithm \tour.}\FigureName{tour}
\end{figure*}

We now proceed to the analysis of our algorithm.  First we show that
the element with the median number of wins (or the element of order
$\ceil{n/2}$ when sorted in increasing order by number of wins)
must incur a significant
number of losses.  We in fact show that it must also incur a
significant number of wins, but we will not need this latter fact
until presenting our selection and sorting algorithms.

\begin{lemma}\LemmaName{pivots-exist}
In a \roundRobin tournament on $n$ elements, the element with the
median number of wins has
at least $m$ wins and at least $m$ losses for $m = \ceil{(\ceil{n/2}-1)/2} \geq \ceil{n/5}$.
\end{lemma}
\begin{proof}
Sort the elements by their total number of wins and let $x'$ be the
median according to the number of wins.
Let $\ell$ denote the number of wins of $x'$. Assume that $n$ is
even. Then the total number of wins for all the elements is at most
$n/2 \cdot \ell   +  n/2 \cdot  (n-1) - {n/2  \choose 2}$ since there
are $n/2$ elements with at most $\ell$ wins and the total number of
wins by the $n/2$ elements that have more wins than $x'$ is at most
$n/2 \cdot (n-1) - {n/2  \choose 2}$. But the total number of wins is
exactly ${n \choose 2}$ and therefore we obtain that $\ell \cdot n/2
\geq {n/2  \choose 2}$, or $\ell \geq (n-2)/4$. The bound on the number
of losses is obtained in the same way. If $n$ is odd then this
argument gives a bound of $(\ceil{n/2}-1)/2$.
\end{proof}

\begin{lemma}\LemmaName{TOURNAMENT}
Let $s = 15(c+2)+1, C$ be as in \Figure{tour}, and let $W$ be the final set output in \Figure{tour}. Then for any integer choice of $c\ge 1$:
\begin{itemize}
\item $|W|\le (c+1)/(C\ln 2)\cdot n^{1/3}\ln^2 n$.
\item $W$ is a $1$-max-set with probability at least $1 - n^{-c}$.
\item The algorithm \tour makes at most $(s-1)n$ comparisons. 
\end{itemize}
\end{lemma}
\begin{proof}
By \Lemma{pivots-exist} the
element with the median number of wins has at least $(s-2)/4$ losses
during each \roundRobin tournament in Step \ref{loser}, and thus the
fraction of elements that survive from one iteration to the next in
Step \ref{last-step} is at most $1/2$.
Therefore, the number of iterations is at most $\ceil{\log_2 n}$.
In each
iteration we sample $((c+1)/C)\cdot n^{1/3}\ln n$ elements, and thus the size of
the output $W$ is as claimed.  Also, the total number of
comparisons is at most $\sum_{i=0}^\infty \binom{s}{2}\cdot (n/2^i)/s
= (s-1)n$, where $n/2^i$ is an upper bound on $|N_i|$.

We now show that $W$ is a $1$-max-set with probability at least
$1-n^{-c}$. We say that an iteration $i$ is {\em good} if either $x^*$
advances to $N_{i+1}$, or $W$ contains a $1$-max
at the end of round $i$. We then show that for any
iteration $i$, conditioned on the event that iterations $1,\ldots,i-1$
were good we have that $i$ is good with probability
$1-1/n^{c+1}$. The lemma would then follow by a union bound over all
iterations $i$.

Now consider an iteration $i$ where $1,\ldots,i-1$ were good.
Then either $W$ already contains a $1$-max, in which case $i$ is good,
or $W$ does not contain a $1$-max but $x^*\in N_i$. Let us focus on
this latter case.  Define $n_i = |N_i|$. Let
$\mathcal{Q}_i$ be
the event that the number of elements in $N_i$ with value at least
$x^*-1$ is at least $C\alpha n_i$ for $\alpha = n^{-1/3}$. We show a
dichotomy: either
$\mathcal{Q}_i$ holds, in which case $1$-max is sampled into $W$ with
probability $1 - 1/n^{c+1}$, or $\mathcal{Q}_i$ does not hold, in
which case $x^*\in N_{i+1}$ with probability $1 - n^{c+1}$.

To show the dichotomy, let us first assume $\mathcal{Q}_i$
holds. Then, the probability we do not sample a $1$-max into $W$ is at
most
$$ \left(1 - \frac{C \alpha n_i}{n_i}\right)^{((c+1)/C) n^{1/3}\ln n} \le
\left(1 - \frac{C}{n^{1/3}}\right)^{((c+1)/C) n^{1/3}\ln n} \le
e^{-(c+1)\ln n} = \frac{1}{n^{c+1}} .$$

Now let us assume that $\mathcal{Q}_i$ does not hold. Then for $x^*$
to not advance to $N_{i+1}$ we must have that at least $s/5$
$1$-maxes were placed into the same set as $x^*$ in iteration
$i$. Using that $(a/b)^b \le \binom{a}{b} \le (ea/b)^b$, this happens
with probability at most
\begin{align*}
\frac{\binom{C \alpha n_i}{s/5}\cdot \binom{n_i}{4s/5}}{\binom{n_i}{s-1}} &<
n_i \cdot \frac{\binom{C \alpha n_i}{s/5}\cdot
  \binom{n_i}{4s/5}}{\binom{n_i}{s}}\\
&{}\le n_i \cdot \frac{e^s\cdot (5 C \alpha n_i/s)^{s/5}\cdot (5
  C n_i/(4s))^{4s/5}}{(n_i/s)^s}\\
&{}\le n_i \cdot (e\cdot C^{1/5}\cdot 5^{1/5}\cdot (5/4)^{4/5} \cdot
n^{-1/15})^s\\
&{}= n^{-s/15 + 1} ,
\end{align*}
which is at most $n^{-(c+1)}$ by our choice of $s$ and $C$.
\end{proof}

\begin{proofof}{\Theorem{exists-rand-algo}}
Run the algorithm in \Figure{tour}. By \Lemma{TOURNAMENT}, the output
$W$ is $1$-max-set with probability at least $1-n^{-c}$.
Conditioned on this event, a $2$-max of $W$ is thus a $3$-max of the
entire original input. A $2$-max of $W$ can be found via a \roundRobin
tournament by \Theorem{roundrobin} using $\binom{|W|}{2} < |W|^2/2$
comparisons. The total number of comparisons is thus the sum of
comparisons made in \Figure{tour}, and in the final round robin
tournament, which gives the bound claimed in the theorem statement.
\end{proofof}

\section{Sorting and Selection}\SectionName{sorting}
We now consider the problems of sorting and selection. We first present an algorithm \ksort{2} which sorts with error 2 using $O(n^{3/2})$ comparisons (and, in particular can be used for selection with error 2). We then describe the selection and sorting algorithms for general error $k$. We start by formally defining what is meant by selecting an element of certain order with error.

\begin{definition}
Element $x_j$ in the set $X = \{x_1,\ldots,x_n\}$ is of {\em $k$-order $i$} if
there exists a partition $S_1,S_2$ of $X\setminus\{x_j\}$ such that $|S_1|=i-1$,
and $S_1 \cup \{x_j\} \le_k S_2 \cup \{x_j\}$. A {\em $k$-median} is an element of $k$-order $\ceil{n/2}$.
\end{definition}

Our error 2 sorting algorithm is based on modifying \kmaxfind{2} so that the $x$ found in Step 2(a) of \Figure{A2}
is used as a pivot.  We then compare this $x$ against all elements and pivot into two sets, recursively sort each, then concatenate. More formally, the algorithm \ksort{2} works as follows. If $n\leq 64$ then we just perform a \roundRobin tournament on the elements. Otherwise let $s = s(n) = \sqrt{2n}$. We choose some $s$ elements and perform a \roundRobin tournament on them. Now let $x$ be an element with the median number of wins. We compare $x$ to all elements and let $S_1$
be the set of elements that lost to $x$ and $S_2$ be the set
of all elements that defeated $x$. Then we recursively sort $S_1$ and $S_2$ then output sorted $S_1$, then $x$, and then sorted $S_2$. Note that any $(y,y')\in(S_1\cup\{x\})\times S_2$ satisfies $y'\ge_2 y$ since
$y'\ge_1 x$ and $x\ge_1 y$. Correctness follows by induction from \Theorem{roundrobin}. We prove the following bound on the number of comparisons used by \ksort{2}.

\begin{theorem}\TheoremName{B2}
There is a deterministic sorting algorithm \ksort{2} with error $2$ that
requires at most $4\cdot n^{3/2}$ comparisons.
\end{theorem}
\begin{proof}
Let $g(n)$ be the worst-case number of comparisons required by \ksort{2} that we described above.
We claim $g(n) \leq 4 \cdot n^{3/2}$. For $n\le 64$ we perform a \roundRobin tournament and thus $g(n) = \binom{n}{2}$ in this case, which is less than $4\cdot n^{3/2}$ for $n \le 64$. Now we consider $n > 64$. Let $t = |S_2|$. The number of comparisons used in a recursive call is at most $s(s-1)/2 + n - s + g(n-t-1) + g(t)$. By our inductive assumption this is at most $s(s-1)/2 + n - s + 4 ((n-t-1)^{3/2} + t^{3/2})$. Without loss of generality we can assume that $t \leq (n-1)/2$ and then obtain that this function is maximized when $t$ is the smallest possible (via a straightforward analysis of the derivative). By \Lemma{pivots-exist}, $x$ has at least $(s-2)/4$ wins and $(s-2)/4$ losses and therefore $t \geq (s-2)/4$.
Now we observe that $(n-t-1)^{3/2} \leq n^{3/2} - \frac{3}{2}\sqrt{n} (t+1)$ (to verify it is sufficient to square both sides and use the fact that $t\leq n$).
Hence
\begin{align*}
g(n) & \leq s^2/2 + n - 3s/2 + 4 (n^{3/2} - \frac{3}{2}\sqrt{n} (t+1) + t^{3/2})\\& \leq
4 \cdot n^{3/2} + 2n - 2\sqrt{2n} - \frac{3 \sqrt{n}(\sqrt{2n}+2)}{2} + n^{3/4} \\ & =
4 \cdot n^{3/2} - (3/\sqrt{2}-2)n - (3 + 3/\sqrt{2})\sqrt{n} + n^{3/4} < 4 \cdot n^{3/2}.
\end{align*}
The last line of the equation follows from the following application of the inequality of arithmetic and geometric means $$(3/\sqrt{2}-2)n + (3 + 3/\sqrt{2})\sqrt{n} \geq 2 \sqrt{(3/\sqrt{2}-2)(3 + 3/\sqrt{2})} n^{3/4} > n^{3/4}\ .$$
This finishes the proof of the induction step.
\end{proof}

At a high level our algorithm for $k$-order selection is similar to
the classical selection algorithm of Blum et al.~\cite{BFPRT73}, in that in each step we try to find a pivot that allows us to recurse on a problem of geometrically decreasing
size.  In our scenario though, a good pivot must have an additional property which we now define.
\begin{definition}
\DefinitionName{k-pivot}
Element $x_j$ in the set $X = \{x_1,\ldots,x_n\}$ is a $k$-pivot for $m$ elements if
there exist disjoint sets $S_w \subset X\setminus\{x_j\}$ ({\em winning set}) and $S_l \subset X\setminus\{x_j\}$ ({\em losing set}), such that $|S_w|=|S_l|=m$, $x_{\ell} \le_k x_j$ for all $\ell \in S_l$, and $x_{\ell} \ge_k x_j$ for all $\ell\in S_w$.
\end{definition}

In order to use an element as a pivot in our algorithm it must be a $(k-1)$-pivot. We construct such a pivot via a recursive algorithm that given $n$ elements and a number $k$ constructs a $k$-pivot for at least $n/(5\cdot 2^{k-1})$ elements. For $k=1$, \Lemma{pivots-exist} gives the desired algorithm. The general algorithm is effectively a recursive application of \Lemma{pivots-exist} and is described below.
\begin{figure*}
\begin{center}
\fbox{
\parbox{6in}{
\underline{Algorithm \kpivot{k}}:
\texttt{// Given a set of $n \geq 3$ elements $X$, returns a $k$-pivot for $m \geq n/(5 \cdot 2^{k-1})$ elements and corresponding losing and winning sets (as described in \Definition{k-pivot}).}
\begin{enumerate}
\item \textbf{if} $k = 1$ or $n \leq 215$ \texttt{// Base case}
\begin{enumerate}
\item Perform a \roundRobin tournament on $X$.
\item Let $y$ be the element with the median number of wins.
\item Let $m$ be the smaller of the number of wins and the number of losses of $y$.
\item Let $S_l$ be a set of $m$ elements that lost to $y$ and $S_w$ be a set of $m$ elements that defeated $y$.
\end{enumerate}
\item \textbf{else}
\begin{enumerate}
\item Set $s \varass \ceil{3 \cdot n^{\frac{1}{2^{k}-1}}}$ and set $s' \varass \ceil{(\ceil{s/2}-1)/2}$.
\item Initialize $T \varass X$, $i \varass 0$.
\item \textbf{while} $|T| \geq s$
\begin{enumerate}
\item $i \varass i+1$.
\item Let $X_i$ be a set of $s$ arbitrary elements from $T$.
\item Perform a \roundRobin tournament on $X_i$.
\item Set $y_i$ to be the element with the median number of wins.
\item Set $S_{i,l}$ to be a set of any $s'$ elements that lost to $y$ and $S_{i,w}$ a set of $s'$ elements that defeated $y_i$.
\item Update $T \varass T\setminus(S_{i,l} \cup S_{i,w} \cup \{y_i\})$.
\end{enumerate}
\item Set $t\varass i$ and $Y \varass \{y_1,y_2,\ldots,y_t\}$.
\item Recursively call \kpivot{(k-1)} on $Y$ and let $y$, $S'_{l}$ and $S'_{w}$ be the $(k-1)$-pivot and the sets returned.
\item Set $S_l \varass S'_l \cup \bigcup_{y_i \in S'_l} S_{i,l}$ and $S_w \varass S'_w \cup \bigcup_{y_i \in S'_w} S_{i,w}$.
\end{enumerate}
\item Return $y$, $S_l$ and $S_w$.
\end{enumerate}
}
}
\end{center}
\caption{The algorithm \kpivot{k} for finding a $k$-pivot for at least $n/(5 \cdot 2^{k-1})$ elements}\FigureName{Dk}
\end{figure*}
Our algorithm for selecting an element of $k$-order $i$ is in \Figure{Ck}.

\begin{figure*}
\begin{center}
\fbox{
\parbox{6.1in} {
\underline{Algorithm \kselect{k}}:
\texttt{// Given a set $X$ of $n \geq 3$ elements, returns an element of $k$-order $i$ in $X$}
\begin{enumerate}
     \item \textbf{if} $k\le 2$ or $n \leq 8$, sort elements using \ksort{2} then return the
       element with index $i$.
     \item \textbf{else}
     \begin{enumerate}
         \item Set $s \varass \floor{n^{1-2^{1-k}}}$ and let $T$ be a set of any $s$ elements from $X$.
         \item Call \kpivot{(k-1)} on $T$ and let $y$, $S_l$ and $S_w$ be the pivot and the sets returned.
         \item Compare $y$ with each of the other $n-1$ elements.
         \item \textbf{if} $y$ defeats at least $(n-1)/2$ elements
         \begin{enumerate}
             \item Set $X_1$ to be the set of all elements that $y$ defeats and are not in $S_w$.
             \item Set $X_2 = X \setminus (X_1\cup \{y\})$.
         \end{enumerate}
         \item \textbf{else} // the symmetric case
         \begin{enumerate}
             \item Set $X_2$ to be the set of all elements that $y$ lost to and are not in $S_l$.
             \item Set $X_1 = X \setminus (X_2\cup \{y\})$.
         \end{enumerate}
         \item \textbf{if} $|X_1| = i-1$ return $y$.
         \item \textbf{else if} $i \le |X_1|$, return the output of \kselect{k} for an element of
               $k$-order $i$ in $X_1$.
         \item \textbf{else} return the output of \kselect{k} for an element of $k$-order $(i -  |X_1| - 1)$ in $X_2$.
    \end{enumerate}
\end{enumerate}
}}
\end{center}
\caption{The algorithm \kselect{k}.}\FigureName{Ck}
\end{figure*}


We claim that algorithm \kselect{k} finds an element of $k$-order $i$ using at most $O(2^k \cdot n^{1+1/2^{k-1}})$ comparisons. To prove this, we first analyze the algorithm \kpivot{k}.
\begin{lemma}\LemmaName{k-pivot}
For any $1\le k\le \log\log n$, given $n$ elements the
deterministic algorithm \kpivot{k} (see \Figure{Dk}) finds a $k$-pivot for $m \geq n/(5 \cdot 2^{k-1})$ elements and corresponding losing set $S_l$ and winning set $S_w$ using at most $9 \cdot n^{1+1/(2^k-1)} + c_n$ comparisons, where $c_n = \min\{{n \choose 2}, {215 \choose 2}\}$.
\end{lemma}
\begin{proof}
We prove that for any $1 \le k \le \log\log n$, \kpivot{k} uses at most $9 \cdot n^{1+1/(2^k-1)} + c_n$ comparisons and finds a $k$-pivot for $m \geq n/(5 \cdot 2^{k-1})$ elements by induction. First, if $k=1$ or $n \leq 215$ then by \Lemma{pivots-exist}, it holds for \kpivot{1} since $m \geq n/5$ and the total number of comparisons $n(n-1)/2 \leq 9 \cdot n^{1+1/(2^k-1)} + c_n$.

We now prove the bound on the error in the general case when $k \geq 2$ and $n \geq 216$. Let $y$ be the element returned by \kpivot{k}. By the inductive hypothesis, for every $v \in S_l$, if $v \in S'_l$ then $v \leq_{k-1} y$. Otherwise, when $v \in S_{i,l}$ for some $y_i \in S'_l$ we get that $v \leq_1 y_i$ and $y_i \leq_{k-1} y$. This implies that in both cases $v \leq_k y$. Similarly, for every $v \in S_w$, $v \geq_k y$.

Next we prove the bound on $m$. The algorithm performs a \roundRobin on $s$ elements until at most $s-1$ elements are left and each such round eliminates exactly $2s'+1$ elements. Therefore the number of rounds $t$ is at least $\geq \ceil{(n-s+1)/(2s'+1)}$. By the inductive hypothesis, $$m' = |S'_{l}|= |S'_{w}| \geq t/(5\cdot 2^{k-2}) \geq\frac{n-s+1} {5\cdot 2^{k-2} \cdot (2s'+1)} \ .$$
Now, $$m = m' \cdot (s'+1) \geq \frac{n-s+1}{2s'+1} \cdot \frac{s'+1}{5 \cdot 2^{k-2}} = \frac{n}{5 \cdot 2^{k-1}} + \frac{n - 2(s-1)(s'+1)}{5 \cdot 2^{k-1} \cdot (2s'+1)}\ .$$ If $s \leq 19$ then $n \geq 216$ implies that $n - 2(s-1)(s'+1) \geq 0$. Otherwise (for $s \geq 20$), $k \geq 2$ implies that $s \leq \ceil{3 \cdot n^{1/3}}$ and therefore, $n \geq ((s-1)/3)^3$. By definition of $s'$, $s'= \ceil{(\ceil{s/2}-1)/2} \leq (s+1)/4$. Therefore, for $s\geq 20$, $$n - 2(s-1)(s'+1) \geq \left(\frac{s-1}{3}\right)^3 - \frac{s^2+4s+5}{2}> 0 .$$
Hence $m = |S_{l}|= |S_{w}| \geq \frac{n}{5 \cdot 2^{k-1}}$.


Finally, we prove the bound on the total number of comparisons used by \kpivot{k} when $k \geq 2$ and $n \geq 216$ as follows:
\begin{itemize}
\item $t \leq \floor{n/(2s'+1)}$ invocations of \roundRobin on $s$ elements. By definition $2s'+1 \geq s/2$ and therefore this step requires at most $\floor{2n/s} s (s-1)/2 \leq n (s-1) \leq 3 \cdot n^{1+1/(2^k-1)}$ comparisons.
\item The invocation of \kpivot{(k-1)} on $t\leq \floor{2n/s}$ elements. By our inductive hypothesis, this requires $9 t^{1+1/(2^{k-1}-1)} + c_t \leq 9(2 \cdot n^{1-1/(2^k-1)}/3)^{1+1/(2^{k-1}-1)} + c_n \leq 6 \cdot n^{1+1/(2^k-1)} + c_n$.
\end{itemize}
Altogether the number of comparisons is at most $9 \cdot n^{1+1/(2^k-1)} +c_n$, as was claimed.
\end{proof}

We are now ready to state and prove our bounds for \kselect{k} formally.
\begin{theorem}\TheoremName{k-select}
For any $2\le k\le \log\log n$, there
is a deterministic algorithm \kselect{k} which, given $n$ elements and $i\in [n]$, finds an element of $k$-order $i$ using at most $25 \cdot 2^{k-1} n^{1+2^{1-k}} + 5 \cdot 2^{2k-3}  n^{2^{1-k}} c_n$ comparisons, where $c_n = \min\{{n \choose 2}, {215 \choose 2}\}$ (as defined in \Lemma{k-pivot}).
\end{theorem}
\begin{proof}
We prove the claim by induction on $n$. For $n \leq 8$ we use \ksort{2} which sorts with error 2. An element $i$ in such a sorting has 2-order $i$. Also, according to \Theorem{B2}, the algorithm uses at most $4n^{3/2} \leq 25 \cdot 2^{k-1} \cdot n^{1+2^{1-k}}$ comparisons.

We now consider the general case when $n \geq 9$ and $k \geq 3$. If $y$ defeats at least $(n-1)/2$ elements then for every element $z_1 \in X_1$, $z_1 \leq_1 y$ and, by the properties of the $(k-1)$-pivot for every element $z_2$ in $X_2$, $y  \leq_{k-1} z_2$. In particular, $z_1 \leq_k z_2$. Now, if $|X_1| = i-1$ then $X_1$ and $X_2$ form  a partition of $X\setminus\{y\}$ showing that $y$ is an element of $k$-order $i$. If $|X_1| > i-1$ then let $y'$ be the $k$-order $i$ element in $X_1$ returned by the recursive call to \kselect{k}. There exists a partition of $X_1$ into sets $S'_1$ and $S'_2$ showing that $y'$ is an element of $k$-order $i$ in $X_1$. We set $S_1 = S'_1$ and $S_2 = S'_2 \cup \{y\} \cup X_2$. First, by the definition of $y'$, for every $z_1 \in S_1 \cup \{y'\}$, and $z'_2 \in S'_2 \cup \{y'\}$ we have $z_1 \leq_k z'_2$. Now, by our choice of $y$ we know that for $z_1 \in S_1 \cup \{y'\}$ and $z_2 \in X_2 \cup \{y\}$ we have $z_1 \leq_k z_2$. Hence $S_1 \cup \{y'\} \leq_k S_2 \cup \{y'\}$, that is, $y'$ is an element of $k$-order $i$. The case when $|X_1| < i-1$ and the symmetric case when $y$ lost to at least $(n-1)/2$ are analogous. This proves that \kselect{k} returns an element of $k$-order $i$.

Finally, in the general case, the number of comparisons \kselect{k} uses is as follows.
\begin{itemize}
\item Call to \kpivot{(k-1)} on $s = \floor{n^{1-2^{1-k}}} \geq \floor{\sqrt{n}} \geq 3$ elements. \Lemma{k-pivot} implies that this step uses at most $9 s^{1+1/(2^{k-1}-1)} + c_s \leq 9 n +c_n$ comparisons.
\item Comparison of the pivot with all the elements uses at most $n-1$ comparisons.
\item The invocation of \kselect{k} on one of $X_1$ and $X_2$. By the definition of pivot, $|S_w| = |S_l| \leq (s-1)/2 \leq (n-2)/4$. Hence we observe that if $y$ defeated at least $(n-1)/2$  elements then $|X_1| \geq (n-1)/2 - |S_w| \geq (n-2)/4 \geq |S_l|$. And by the definition of $X_2$, $|X_2| \geq |S_w|$. Similarly, if $y$ lost to at least $(n-1)/2$ elements then $|X_1| \geq |S_l|$ and $|X_2| \geq |S_w|$. Assume, without loss of generality, that $|X_1| \geq |X_2|$ and let $\alpha = (|X_2|+1)/n$.
     By the properties of the $(k-1)$-pivot, $$|X_2| \geq s/(5\cdot 2^{k-2}) \geq \floor{n^{1-2^{1-k}}}/(5\cdot 2^{k-2}) \geq n^{1-2^{1-k}}/(5\cdot 2^{k-2}) - 1, $$ or $\alpha \geq n^{-2^{1-k}}/(5\cdot 2^{k-2})$.
     The number of comparisons is maximized when \kselect{k} is executed on $X_1$ which has size $n-|X_2|-1 = n-\alpha n$ and, by our inductive hypothesis, this step can be done using $N$ comparisons for \alequn{N & \leq  25 \cdot 2^{k-1} \left(n - \alpha n\right)^{1+2^{1-k}} + 5 \cdot 2^{2k-3} \cdot \left(n - \alpha n\right)^{2^{1-k}} \cdot c_n \\
    \leq & 25\cdot 2^{k-1} \cdot n^{1+2^{1-k}} \left(1-\alpha\right)^{1+2^{1-k}} + 5 \cdot 2^{2k-3} \cdot n^{2^{1-k}} \cdot \left(1 - \alpha\right)^{2^{1-k}} \cdot c_n \\
    \leq & 25 \cdot 2^{k-1} \cdot n^{1+2^{1-k}} (1-\alpha) + 5 \cdot 2^{2k-3} \cdot n^{2^{1-k}} \cdot (1 - 2^{1-k} \cdot \alpha) \cdot c_n\\
     = & 25  \cdot 2^{k-1} \cdot n^{1+2^{1-k}} + 5 \cdot 2^{2k-3} \cdot n^{2^{1-k}} c_n - 10 \cdot n - c_n}
\end{itemize}
Therefore altogether the number of comparisons is at most $25 \cdot 2^{k-1} \cdot n^{1+2^{1-k}} + 5 \cdot 2^{2k-3} \cdot n^{2^{1-k}} \cdot c_n$.
\end{proof}

We now show that with a small change our selection algorithm can be used to produce a complete sorting with error $k$. Namely, instead of running recursively on one of the subsets $X_1$ and $X_2$, we recursively sort each partition, then concatenate the sorted results in the order \ksort{k}$(X_1), y,$\ksort{k}$(X_2)$. As expected, in the base case (when $n \leq 8$) we just output the result of \ksort{2}. We claim that the resulting algorithm, to which we refer to as \ksort{k}, has the following bounds on the number of comparisons.

\begin{theorem}\TheoremName{k-sort}
For any $2\le k\le \log\log n$, there
is a deterministic algorithm \ksort{k} which given a set $X$ of $n$ elements, sorts the elements of $X$ with error $k$ and uses at most $7 \cdot 2^{2k} \cdot n^{1+2^{1-k}} + n \cdot c_n$ comparisons, where $c_n = \min\{{n \choose 2}, {215 \choose 2}\}$ (as defined in \Lemma{k-pivot}).
\end{theorem}
\begin{proof}
As before, we prove the claim by induction on $n$. For $n \leq 8$ we use \ksort{2} which produces sorting with error 2 and gives a suitable bound on the number of comparisons.

We now consider the general case ($n \geq 9$ and $k \geq 3$). As in the case of selection, it is easy to see that
the algorithm sorts $X$ with error $k$. However the bound on the number of comparisons is different from the one we gave for \kselect{k} since now the algorithm is called recursively on both $X_1$ and $X_2$. As before, the call to \kpivot{(k-1)} on $s = \floor{n^{1-2^{1-k}}}$ elements uses at most $9 s^{1+1/(2^{k-1}-1)} + c_n \leq 9 n +c_n$ comparisons and there are at most $n-1$ comparisons of the pivot with all the other elements. We again assume, without loss of generality, that $|X_1| \geq |X_2|$ and denote $\alpha = (|X_2|+1)/n$. By our inductive hypothesis, the number of comparisons $N$ used for the recursive calls is bounded by
\alequ{N &\leq 7 \cdot 2^{2k} \left((n - \alpha n)^{1+2^{1-k}} + (\alpha n)^{1+2^{1-k}}\right) + c_n ((n - \alpha n) + (\alpha n -1)) \nonumber \\
     & = 7 \cdot 2^{2k} \cdot n^{1+2^{1-k}} \left((1-\alpha)^{1+2^{1-k}} + \alpha^{1+2^{1-k}}\right) + (n-1)c_n \ . \label{eq:bound-k-sort}}
     By differentiating the expression $(1-\alpha)^{1+2^{1-k}} + \alpha^{1+2^{1-k}}$ as a function of $\alpha$ we obtain that it is monotonically decreasing in the interval $[0,1/2]$ and hence its minimum is attained when $\alpha$ is the smallest. Therefore we can use the lower bound $\alpha \geq n^{-2^{1-k}}/(5\cdot 2^{k-2})$ to conclude that
     \alequn{ (1-\alpha)^{1+2^{1-k}} + \alpha^{1+2^{1-k}} &\leq  (1-\alpha)(1-\alpha)^{2^{1-k}}  + \alpha \leq (1-\alpha)(1-2^{1-k} \cdot \alpha)  + \alpha \\
      & = 1 - (1-\alpha)\cdot 2^{1-k} \alpha \leq 1 - \frac{9}{10}\cdot 2^{1-k} n^{-2^{1-k}}/(5\cdot 2^{k-2}) \\
      & = 1 - \frac{9}{25} \cdot 2^{2-2k} n^{-2^{1-k}}.}
      By substituting this bound into equation (\ref{eq:bound-k-sort}) we obtain that
      $$N \leq 7 \cdot 2^{2k} \cdot n^{1+2^{1-k}} \left(1 - \frac{9}{25} 2^{2-2k} n^{-2^{1-k}}\right) + (n-1) c_n = 7 \cdot 2^{2k} \cdot n^{1+2^{1-k}} + n \cdot c_n - 10n - c_n\ .$$
Therefore altogether the number of comparisons used by \ksort{k} is at most $7 \cdot 2^{2k} \cdot n^{1+2^{1-k}} + n \cdot c_n$.
\end{proof}
An immediate corollary of \Theorem{k-sort} is that it is possible to achieve error of no more than $\log\log n$ in close to optimal time.
\begin{corollary}\CorollaryName{loglog-sort}
There is a sorting algorithm using $O(n\log^2 n)$ comparisons with error of at most $\log\log n$.
\end{corollary}

\section{Lower Bounds}\SectionName{lb}
Here we prove lower bounds against deterministic max-finding, sorting,
and selection algorithms. In particular, we show that
\Theorem{upper-bound}, \Theorem{k-select} and \Theorem{k-sort} achieve almost optimal
trade-off between error and number of comparisons.

Our proof is based on the analysis of the comparison graph, or the directed graph on all elements in which an edge $(x_i,x_j)$ is present whenever a comparison between $x_i$ and $x_j$ was made and its imprecise outcome was $``x_i \geq x_j"$. We show that one can only conclude that $x_i \geq_k x_j$ if this graph has a path of length at most $k$ from $x_i$ to $x_j$. The existence of short paths from an element to numerous other elements (such as when the element is a $k$-max) is only possible when there are many vertices with large out-degree. Following this intuition we define an oracle that when comparing two elements always responds that the one with the smaller out-degree is larger than the one with the larger out-degree. Such an oracle will ensure that a large number of comparisons needs to be made in order to obtain a sufficient number of vertices with high out-degree. We also show that the responses of the oracle can be seen as derived from actual values defined using the resulting comparison graph.

\begin{lemma}
\LemmaName{lb-lemma}
Suppose a deterministic algorithm $A$ upon given $n$ elements
guarantees that after
$m$ comparisons it can list $r$ elements,
each of which
is guaranteed to be $k$-greater than at least $q$ elements.
Then $m = \Omega(\max\{q^{1 + 1/(2^k-1)},
q\cdot r^{1/(2^{k-1})}\})$.
\end{lemma}
\begin{proof}
To create a worst case input we first define a strategy for the comparator and later choose values for the
elements which are consistent with the given answers, while maximizing the error of the
algorithm.

Let $G_t$ be the comparison graph at time $t$. That is, $G_t$ is a
digraph whose vertices are the $x_i$ and which
contains the directed edge $(x_i, x_j)$ if and only if before time $t$
a comparison between $x_i$ and $x_j$ has been made, and the
comparator has responded with ``$x_i \geq x_j$''. We denote the out-degree
of $x_i$ in $G_t$ by $d_t(x_i)$. Assume that at
time $t$ the algorithm wants to compare some $x_i$ and $x_{j}$.
If $d_t(x_i) \geq d_t(x_j)$ then the comparator responds with ``$x_j
\geq x_i$'', and it responds with ``$x_i \geq x_j$'' otherwise.
(The response is arbitrary when $d_t(x_i) = d_t(x_j)$.) Let $x$ be an
element that is declared by $A$ to be $k$-greater than at least $q$
elements.

Let $\ell_i = \dist(x,x_i)$, where $\dist$ gives
the length of the shortest (directed) path in the
final graph $G_{m}$. If no such path exists, we set $\ell_i =
n$. After the algorithm is done, we define  $\val(x_i) = \ell_i$. We
first claim that the values are consistent with the responses of the comparator. If
for some pair of elements $x_i,x_j$ the comparator has responded with
``$x_i \geq x_j$'', then $G_{m}$ contains edge $(x_i,x_j)$. This implies
that for any $x$, $\dist(x,x_j) \leq \dist(x,x_i) + 1$, or $\ell_i \geq
\ell_j - 1$.  Therefore the answer ``$x_i \geq x_j$" is consistent with
the given values.

Consider the nodes $x_i$ that $x$ can reach via a path of length at most
$k$.  These are exactly the elements $k$-smaller than $x$, and
thus there must be at least $q$ of them.
For $i \leq k$ let $S_i = \{ x_j\ |\ \ell_j = i\}$ and $s_i = |S_i|$.
We claim that for every $i \in [k]$,
$m \geq s_i^2/(2s_{i-1}) -
s_i/2$.
For a node $u \in S_i$, let $\pred(u)$ denote some node in
$S_{i-1}$ such that the edge $(\pred(u),u)$ is in the graph. For a
node $v \in S_{i-1}$, let $S_{i,v} = \{u \in S_i \ |\ v =
\pred(u)\}$. Note that each node $u \in S_i$ has a single designated $\pred(u) \in S_{i-1}$ and hence $S_{i,v}$'s are disjoint.
 Further, let $\outd(\pred(u),u)$ be the out-degree of
$\pred(u)$ when the comparison between $\pred(u)$ and $u$ was made (as
a result of which the edge was added to $G_m$). Note that for any
distinct nodes $u,u' \in S_{i,v}$, $\outd(v,u) \neq \outd(v,u')$ since the
out-degree of $v$ grows each time an edge to a node in $S_{i,v}$  is
added. This implies that $$\sum_{u \in  S_{i,v}} \outd(v,u) \geq \sum_{d
  \leq |S_{i,v}| - 1} d = |S_{i,v}| (|S_{i,v}| - 1)/2\ .$$ By the
definition of our comparator, for every $u \in S_i$,
$d_m(u) \geq \outd(\pred(u),u)$. This implies that
$$m \geq \sum_{v \in
  S_{i-1}} \sum_{u \in S_{i,v}} d_m(u) \geq \sum_{v \in S_{i-1}}
\frac{|S_{i,v}| (|S_{i,v}| - 1)}{2} = \frac{\sum_{v \in S_{i-1}}
  |S_{i,v}|^2 - |S_i|}{2}\ .$$
Using the inequality between the
quadratic and arithmetic means, $$\sum_{v \in S_{i-1}}
|S_{i,v}|^2 \geq \left(\sum_{v \in S_{i-1}} |S_{i,v}|\right)^2 /|S_{i-1}| =
s_i^2/s_{i-1}. $$ This implies that $m \geq \frac{s_i^2}{2s_{i-1}} -
\frac{s_i}{2}$.

We can therefore conclude that $s_i \leq \sqrt{(2m + s_i) s_{i-1}}
\leq \sqrt{3m s_{i-1}}$ since $s_i \leq n \leq m$. By applying this
inequality and using the fact that $s_0 =1$ we obtain
that $s_1^2/3\le m$ and $s_i \leq 3m \cdot (3m/s_1)^{2^{-(i-1)}}$
for $i>1$. Since $\sum_{i \leq k} s_i \ge q+1$, we thus find
that $q \leq 12 \cdot m \cdot (3m/s_1)^{2^{-(k-1)}}$.
This holds since either
\begin{enumerate}
\item $(3m/s_1)^{2^{-(k-1)}} > 1/2$ and then
  $12 \cdot m \cdot (3m/s_1)^ {2^{-(k-1)}} \geq 6m > n$,
  or
\item $(3m/s_1)^{-2^{-(k-1)}} \le 1/2$ and then
  $$(3m/s_1)^{-2^{-i+1}}/(3m/s_1)^{-2^{-i}} = (3m/s_1)^{-2^{-i}} \leq
  (3m/s_1)^{-2^{-(k-1)}} \leq 1/2$$ for $i\le k-1$, where the
  penultimate inequality holds since $s_1 < 3m$. In this case
\begin{eqnarray*}
  q-s_1 \le \sum_{i=2}^k s_i
  \leq \sum_{i=2}^k (3m)(3m/s_1)^{-2^{-(i-1)}} &\leq&  \sum_{i \leq k}
  2^{i-k}
  (3m)(3m/s_1)^{-2^{-(k-1)}} \\
  &<&
  2(3m)^{1-2^{-(k-1)}}s_1^{2^{-(k-1)}}
\end{eqnarray*}
\end{enumerate}
If $s_1 \ge q/2$, then $m = \Omega(q^2)$ since $m\ge
s_1^2/3 $.  Otherwise we have that
$$m \geq
(q/4)^{1/(1-2^{-(k-1)})}/(3s_1^{1/(2^{(k-1)} - 1)}) ,$$
implying
$$m =
\Omega(\max\{s_1^2,q^{1/(1-2^{-(k-1)})}/s_1^{1/(2^{(k-1)} - 1)}\}) =
\Omega(q^{1 + 1/(2^k-1)}) ,$$
where the final equality can be seen by making the two terms in the
max equal.

Also, note that the choice of $x$ amongst the $r$ elements of the
theorem statement was arbitrary, and that $s_1$ is just the out-degree
of $x$.  Let $s_{\mathrm{min}}$ be the minimum out-degree
amongst the $r$ elements.
Then we trivially have $m \ge r\cdot s_{\mathrm{min}}$. Thus, if
$s_{\mathrm{min}} \ge q/2$ then $m \ge qr/2$, and otherwise
$$m =
\Omega(\max\{r\cdot
s_{\mathrm{min}},q^{1/(1-2^{-(k-1)})}/s_{\mathrm{min}}^{1/(2^{(k-1)} - 1)}\}) =
\Omega(q\cdot r^{1/(2^{k-1})})$$
where the final equality is again seen by making the two terms in
the max equal.
\end{proof}

From here a 
a lower bound for
max-finding by setting $r=1$, $q=n-1$, and for median-finding and
sorting by setting $r=q=n/2$. The sorting lower bound
holds for $k$-order selection of the $i^{th}$ element for any $i =
c\cdot n$ for constant $0 < c < 1$. More generally, selecting an element of order $i \geq n/2$ requires
$\Omega(i \cdot \max\{i^{1/(2^k-1)}, n^{1/(2^{k-1})}\})$ comparisons. For $i \leq n/2$ we obtain the lower bound of $\Omega((n-i) \cdot \max\{(n-i)^{1/(2^k-1)}, n^{1/(2^{k-1})}\})$ by considering the symmetric problem.
\begin{theorem}\TheoremName{maxfind-lb}
Every deterministic max-finding algorithm $A$ with error $k$
requires $\Omega(n^{1 + 1/(2^k-1)})$ comparisons.
\end{theorem}
\Theorem{maxfind-lb} implies that \ksort{k} and \kselect{k} are optimal up to a constant factor for any constant $k$.

\begin{theorem}
Every deterministic algorithm $A$ which $k$-sorts $n$ elements, or
finds an element of $k$-order $i$ for $i=c\cdot n$ with $0<c<1$
a constant,
requires $\Omega(n^{1 + 1/2^{k-1}})$ comparisons.
\end{theorem}


In addition we obtain that \Corollary{linear-time-max} is essentially tight.
\begin{corollary}
Let $A$ be a deterministic max-finding algorithm that makes $O(n)$
comparisons. Then $A$ has error at least $\log\log{n} - O(1)$.
\end{corollary}

\section{Conclusions}
We defined a simple and natural model
of imprecision in a result of a comparison. The model is inspired by both imprecision in human judgement of values and also by bounded but potentially adversarial errors in sporting tournaments. Despite the basic nature of the model and the vast literature on sorting and searching with faulty comparisons we are not aware of any prior efforts to address this type of errors. Our results show that there exist algorithms that are robust to imprecision in comparisons while using substantially fewer comparisons than the na\"ive methods. For deterministic algorithms our problem can equivalently be seen as finding a $k$-king in any tournament graph (or sorting elements so that each element is a $k$-king for vertices of lower order) while minimizing the number of edges checked. Our results generalize previous work on this problem that considered the case of $k=2$ \cite{ShenSW03}.

We note that in most of the results substantially tighter constants can be obtained in the bounds using small modifications of the algorithms, more careful counting
and optimization for small values of $k$ and $n$. This would yield algorithms that improve significantly on the na\"ive approach even for small values of $n$. We made only a modest effort
to improve the constants to make the presentation of the main ideas clearer.

While our lower bounds show that many of the algorithms we give are essentially optimal a number of interesting and natural problems are left open.
\begin{enumerate}
\item What is the complexity of deterministic maximum finding with error $2$? \kmaxfind{2} uses $O(n^{3/2})$ comparisons whereas our lower bound is $\Omega(n^{4/3})$ comparisons. Resolving the case of $k=2$ is likely to lead to closing of the gap for larger error $k$.
\item Can error 2 be achieved by a randomized algorithm using $O(n)$ comparisons? \tour only guarantees error 3.
\item For randomized algorithms it is also natural to consider the expected error of an algorithm. What is the lowest expected error that can be achieved using a randomized algorithm for the tasks considered in this paper? Note that in the example presented for the proof of \Theorem{performance-2}, choosing a random element would give a maximum element with expected error of 1 and this is the best possible in this example.
\item We have not addressed the complexity of randomized sorting with error $k$.
\end{enumerate}


\bibliographystyle{plain}
\bibliography{./allpapers}


\end{document}